\documentclass[preprint,nofootinbib,amsmath,amssymb,amsfonts,aps,prd]{revtex4-1}

\usepackage{setspace}
\usepackage{graphicx}
\usepackage{subfigure}
\usepackage[bookmarks=true, pdfstartview={FitH}, bookmarksopen=true, bookmarksopenlevel=1, linktoc=page]{hyperref}
\usepackage{amsthm}
\usepackage{mathrsfs}
\usepackage{gensymb}
\usepackage{color}

\DeclareMathOperator{\lie}{\pounds}

\newcommand{\be}{\begin{equation}}
\newcommand{\ee}{\end{equation}}
\newcommand{\ba}{\begin{equation}\begin{aligned}}
\newcommand{\ea}{\end{aligned}\end{equation}}

\DeclareMathAlphabet{\mathpzc}{OT1}{pzc}{m}{it}

\usepackage{amsthm}
\let\oldproofname=\proofname
\renewcommand{\proofname}{\rm\bf{\oldproofname}:}
\newtheorem*{theorem}{Theorem}

\newtheorem*{lemma}{Lemma}

\usepackage{mathrsfs}

\begin{document}

\pagestyle{myheadings}

\title{Reflection Symmetry in Higher Dimensional Black Hole Spacetimes}
\author{Joshua S. Schiffrin}
\email{schiffrin@uchicago.edu}
\author{Robert M. Wald}
\email{rmwa@uchicago.edu}
\affiliation{Enrico Fermi Institute and Department of Physics \\
  The University of Chicago \\
  5640 S. Ellis Ave., Chicago, IL 60637, U.S.A.}
\date{\today}

\begin{abstract} 
In 4 spacetime dimensions there is a well known proof that for any asymptotically flat, stationary, and axisymmetric vacuum solution of Einstein's equation there exists a ``$t$-$\phi$'' reflection isometry that reverses the direction of the timelike Killing vector field and the direction of the axial Killing vector field.
However, this proof does not generalize to higher spacetime dimensions.
Here we consider asymptotically flat, stationary, and axisymmetric (i.e., having one or more commuting rotational isometries) black hole spacetimes in vacuum general relativity in $d \geq 4$ spacetime dimensions such that the action of the isometry group is trivial. (Here ``trivial'' means that if the ``axes''---i.e., the points where the axial Killing fields are linearly dependent---are removed, the action of the isometry group is that of a trivial principal fiber bundle. This excludes actions like that found in the Sorkin monopole.) We prove that there exists a ``$t$-$\phi$'' reflection isometry that reverses the direction of the timelike Killing vector field and the direction of each axial Killing vector field. The proof relies in an essential way on the first law of black hole mechanics.
\end{abstract}

\maketitle



\section{Introduction} \label{intro}

A $d$-dimensional spacetime, $(M,g_{ab})$, is said to be {\it stationary and axisymmetric} if its isometry group contains the group 
\be
G \equiv R \times \underbrace{U(1) \times \dots \times U(1)}_{n \text{ times}} \, , 
\label{G}
\ee
where action of $G$ on $M$ is such that the orbits of the $R$-factor are timelike in a suitable region of spacetime (stationarity) and the orbits of each $U(1)$ factor are spacelike (axisymmetry). Equivalently, a stationary-axisymmetric spacetime is one in which there exist Killing vector fields $t^a$, ${\phi_\Lambda}^{a}$ for $\Lambda=1,\dots,n$ that mutually commute, are linearly independent (as Killing fields, but not necessarily as vectors at each point), and are such that the orbits of $t^a$ are timelike and the orbits of each ${\phi_\Lambda}^a$ are spacelike and closed.

For stationary and axisymmetric spacetimes in $4$ spacetime dimensions with one axial Killing field, $\phi^a$, there is a well known result, due to Papapetrou \cite{Papapetrou} and Carter \cite{carter} (see, e.g., section 7.1 of \cite{WaldGR}) that if $g_{ab}$ satisfies the vacuum Einstein equation and $\phi^a$ vanishes at a point, then the timelike Killing field, $t^a$, and the axial Killing field, $\phi^a$, are $2$-surface orthogonal, i.e., the $2$-planes orthogonal to $t^a$ and $\phi^a$ are integrable. This $2$-surface orthogonality allows one to introduce local coordinates $(t,\phi,x^2,x^3)$ adapted to the Killing fields so that the metric components are independent of $(t,\phi)$ and take a ``block diagonal'' form with respect to $(t,\phi)$ and $(x^2, x^3)$. It is then easily seen that the (local) diffeomorphism $(t,\phi) \to (-t, - \phi)$, $(x^2, x^3) \to (x^2, x^3)$ is an isometry. Thus, in $4$ spacetime dimensions, every stationary and axisymmetric solution of the vacuum Einstein equation such that $\phi^a$ vanishes at a point possesses a ``$(t$-$\phi)$-reflection isometry,'' at least locally. In particular, this result holds for all asymptotically flat, stationary, and axisymmetric solutions, since $\phi^a$ must vanish on a ``rotation axis'' in the asymptotically flat region.

However, the standard proof of the surface orthogonality of $t^a$ and $\phi^a$ is highly dependent on the spacetime dimension being $4$. To see this, we note that in $d$ spacetime dimensions, the Frobenius conditions for the local existence of $(d-2)$-dimensional surfaces orthogonal to both $t^a$ and $\phi^a$ can be expressed as
\begin{align}
\phi^{[a}t^b \nabla^c t^{d]}&=0 \label{frobenius1}\\
t^{[a}\phi^b \nabla^c \phi^{d]}&=0 \label{frobenius2}  \,  .
\end{align}
However, only in $4$ dimensions can the first of these equations be written as the scalar equation
\be
\phi^a \omega_a = 0 \, , 
\label{frob4d}
\ee
where
\be
\omega_a \equiv \epsilon_{abcd}t^b \nabla^c t^d
\ee
is the ``twist'' of $t^a$. The second equation, of course, can be similarly expressed. To show that \eqref{frob4d} holds, we use the identity
\be
\nabla_b \left(\phi^a \omega_a\right) = \lie_\phi \omega_b - 2 \phi^a \nabla_{[a} \omega_{b]} \, .
\ee
The first term on the right side of this equation vanishes since $\phi^a$ is a Killing field that commutes with $t^a$. The second term on the right side can be shown to be proportional to the Ricci tensor, and it therefore vanishes by Einstein's equation. Consequently, in $4$-dimensions, we obtain $\nabla_b (\phi^a \omega_a)= 0$. But, by hypothesis, $\phi^a$ vanishes at a point, so \eqref{frob4d}---and equivalently \eqref{frobenius1}---holds. Similarly, in $4$-dimensions, \eqref{frobenius2} holds, thus proving the desired result.

To see why this proof fails in higher dimensions, consider a $5$-dimensional spacetime, but still with just two Killing vector fields $t^a$ and $\phi^a$. The first Frobenius condition \eqref{frobenius1} is now
\be
\phi^a \omega_{ae} = 0 \, ,
\ee
where
\be
\omega_{ae} \equiv \epsilon_{aebcd}t^b \nabla^c t^d \, .
\ee
We again have the identity
\be
d \left(\iota_\phi \boldsymbol\omega \right) = \lie_\phi \boldsymbol\omega - \iota_\phi d\boldsymbol\omega \, ,
\ee
where we have now switched to differential forms notation and $\iota_\phi$ denotes the contraction of $\phi^a$ into the first index of a differential form. It is again true that $\lie_\phi \boldsymbol\omega = 0$ and $d\boldsymbol\omega = 0$, which implies that $d (\iota_\phi \boldsymbol\omega) = 0$. However, the condition $d (\iota_\phi \boldsymbol\omega) = 0$ together with the vanishing of $\iota_\phi \boldsymbol\omega$ at a point (or even, e.g., its vanishing on a higher dimensional surface) does not imply that $\iota_\phi \boldsymbol\omega = 0$, since this quantity is now a one-form, not a scalar, and there are plenty of nonvanishing closed one-forms that vanish at a point.
Thus, in 5 dimensions with one axial Killing vector field, the proof fails.
A direct analog of the proof does work in 5 spacetime dimensions with 2 commuting axial Killing vector fields, and more generally in $d$ spacetime dimensions with $d-3$ commuting axial Killing vector fields \cite{EmparanReall}.
But for asymptotically flat spacetimes in $d>5$ dimensions, one cannot have $d-3$ independent commuting rotations, so in the asymptotically flat case, the above proof or its direct analog works for no cases other than $d=4$ with one axial Killing vector field and $d=5$ with two axial Killing vector fields.

The purpose of this paper is to provide a completely different proof of the surface orthogonality and corresponding reflection symmetry for stationary and axisymmetric vacuum spacetimes. Our proof will hold in arbitrary dimensions with an arbitrary number of commuting axial Killing vector fields. However, our proof will require us to impose the following restrictions that are not needed in the $4$-dimensional proof: (i) We consider only asymptotically flat, black hole spacetimes. (ii) We require the action of the isometry group $G$ on $M$ to be trivial in the sense that $M \setminus \mathcal A = {\mathcal O} \times G$ for some manifold $\mathcal O$, where $\mathcal A$ denotes the set of points (``axes'') where the Killing fields are linearly dependent. Thus, our generalization to higher dimensions of the proof of the existence of a reflection isometry comes at the price of imposing some additional hypotheses.

Our proof parallels the key idea of the proof of Theorem 3.1 of \cite{sudwald} and can be summarized as follows: We start with a maximal (i.e., zero trace of the extrinsic curvature) hypersurface, $\Sigma$, that is asymptotically flat and terminates at the bifurcation surface of the black hole horizon.
(Such a hypersurface exists by the results of \cite{ChruscielWald}.) We show that the axial Killing fields must be tangent to $\Sigma$. A ``$t$-$\phi$'' reflection isometry about $\Sigma$ will exist if and only if (a) the induced metric, $h_{ij}$, of $\Sigma$ is such that there exists a reflection isometry $i_\Sigma : \Sigma \to \Sigma$ that reverses the directions of the axial Killing fields and (b) the extrinsic curvature, $K_{ij}$, of $\Sigma$ reverses sign under the action of $i_\Sigma$. If these conditions hold, then the desired ``$t$-$\phi$'' reflection isometry is obtained by mapping a point $p$ lying at proper time $\tau$ along a normal geodesic starting at $s \in \Sigma$ to the point $q$ lying at proper time $-\tau$ along a normal geodesic starting at $i_\Sigma (s)$. Condition (a) is locally equivalent to the Frobenius conditions holding for the axial Killing fields on $\Sigma$, which, in turn, is equivalent to the vanishing of a curvature $2$-form constructed from the axial Killing fields. Condition (b) is equivalent to the vanishing of the ``polar part'' of the extrinsic curvature. To show that these quantities vanish, we construct a perturbation which scales up the curvature $2$-form and scales up the ``polar part'' of the extrinsic curvature. The trivial action of $G$ (assumption (ii) above) is used here to globally construct the desired perturbation and to ensure that this perturbation is smooth on $\Sigma$. This perturbation satisfies the linearized momentum constraints, but it fails to satisfy the linearized Hamiltonian constraint. We therefore adjust the perturbation with a linearized conformal transformation so that all of the Einstein constraint equations are satisfied. The resulting perturbation can then be easily seen to have $\delta A = \delta J_\Lambda = 0$, where $A$ is the area of the horizon and $J_\Lambda$ is the ADM angular momentum associated with axial Killing vector field ${\phi_\Lambda}^i$. Hence, by the first law of black hole mechanics, the change in ADM mass, $\delta M$, is zero. However, we also prove that $\delta M = 0$ only if the curvature $2$-form and the polar part of $K_{ij}$ vanish. Consequently, we locally obtain the desired reflection isometry. This construction can then be made global using the simply connectedness of the black hole exterior.

In section \ref{assumptions}, we spell out our assumptions about the spacetime and the stationary and axisymmetric symmetries. In section \ref{manifoldoforbits}, we explain how the spatial hypersurface $\Sigma$ (with the axes removed) may be viewed as a principal fiber bundle with group $U(1) \times \dots \times U(1)$. Although our assumptions require this bundle to be trivial, it is very useful to introduce the fiber bundle language to explain how the axial Killing fields (together with the structure obtained from the spatial metric, $h_{ij}$) give rise to a connection on this bundle and to introduce the curvature of this connection. In section \ref{constraintssec}, we write the constraint equations as equations on the ``manifold of orbits,'' i.e., the base space of this fiber bundle. In section \ref{conformalsection}, we construct the perturbation that will be used in our proof. In section \ref{proofsec}, we state and prove our reflection isometry theorem. Finally, we briefly mention some possible generalizations of our results in section \ref{gensec}. In Appendix \ref{axisappendix}, we derive the axis regularity conditions that follow from the smoothness of the metric under our assumption of trivial group action.

Our index notational conventions are as follows. Lower case Latin indices from the early alphabet ($a,b,c,\dots$) will denote abstract spacetime indices, whereas lower case Latin indices from mid-alphabet ($i,j,k,\dots$) will denote abstract spatial indices. However, we will not distinguish notationally between spatial tensors on a hypersurface $\Sigma$ and their projection to the manifold of Killing orbits, $\mathcal O$. 
Lower case Greek indices ($\mu,\nu,\dots$) will denote coordinate labels, components, or basis elements. Upper case Greek indices ($\Lambda, \Gamma, \dots$) will be used to enumerate the axial Killing fields ${\phi_\Lambda}^a$ and corresponding quantities, such as the angular momenta, $J_\Lambda$, and the angular coordinates, $\varphi^\Lambda$. Upper case Latin indices will be used to denote tensors over the abstract Lie algebra, $\mathcal V$, of the group of symmetries. Thus, ${\phi_A}^a$ denotes the collection of Killing fields $\{{\phi_\Lambda}^a\}$ when viewed as a vector field on spacetime that is valued in the dual of $\mathcal V$.

\section{Assumptions}\label{assumptions}

Let $(M,g_{ab})$ be a $d$-dimensional spacetime. As defined in section \ref{intro}, a stationary, axisymmetric spacetime is one in which the isometry group contains a subgroup of the form \eqref{G}, where the orbits of the $R$-factor are timelike in a suitable region of spacetime (stationarity) and the orbits of each $U(1)$ factor are spacelike (axisymmetry). Let $t^a$, ${\phi_\Lambda}^{a}$ for $\Lambda=1,\dots,n$ denote the corresponding Killing fields. For reasons that will become clear when we give our proof, the results of this paper will apply only to stationary and axisymmetric spacetimes such that the action of $G$ is trivial in the following sense: Let $\mathcal A$ denote the set of points at which the orbits of $G$ fail to be $(n+1)$-dimensional; equivalently, $\mathcal A$ is the set of points at which the Killing fields $t^a$, ${\phi_\Lambda}^a$ fail to be linearly independent. Let $\widetilde{M} = M \setminus {\mathcal A}$. Since a non-zero Killing field that vanishes at a point cannot have vanishing derivative at that point, it follows immediately that $\widetilde{M}$ is an open, dense subset of $M$, so, in particular, $\widetilde{M}$ itself is a manifold. We say that {\em the action of $G$ on $M$ is trivial} if $\widetilde{M}$ can be written in the form $\widetilde{M} = {\mathcal O} \times G$, where ${\mathcal O}$ is a $(d-n-1)$-dimensional manifold, such that the action of the isometry group $G$ on $\widetilde{M}$ is given by $\psi_{g'} (p, g) = (p, g'g)$, where $p \in \mathcal O$ and  $g,g' \in G$. We refer to $\mathcal O$ as the manifold of orbits of $G$. 

There are two distinct ways by which the assumption of a trivial action of $G$ could fail. The first is that, even though by construction the Killing fields are linearly independent at each point of $\widetilde{M}$, there could exist a discrete subgroup of $G$ that has fixed points on $\widetilde{M}$; see, e.g., \cite{Orlik} for examples. This would give the ``manifold of orbits'' the structure of an orbifold rather than a manifold. 

The second way our assumption of a trivial $G$ action could fail is for $\widetilde{M}$ to fail to be of the form ${\mathcal O} \times G$. In the language of fiber bundles (see section \ref{manifoldoforbits}), $\widetilde{M}$ would be a principal fiber bundle with structure group $G$, but it would not be a trivial bundle. This type of behavior is displayed in the Sorkin monopole \cite{Sorkin1983,GrossPerry}. An even simpler example is given by the metric\footnote{This metric, of course, is not a solution to Einstein's equation.}
\be\label{sorkinmonopole}
ds^2 = -dt^2 + \left(1 + \frac{xz + yw}{1+r^4}\right)[dx^2 + dy^2 +dz^2 + dw^2] \, ,
\ee
where $r^2 = x^2 + y^2+ z^2 + w^2$. This metric is asymptotically flat and has Killing fields $t^a = (\partial/\partial t)^a$ and
\be
\phi^a = x \left(\frac{\partial}{\partial y}\right)^a - y \left(\frac{\partial}{\partial x}\right)^a + z \left(\frac{\partial}{\partial w}\right)^a
- w \left(\frac{\partial}{\partial z}\right)^a \, ,
\label{ntphi}
\ee
which can be seen to be a sum of rotations in the $x$-$y$ and $z$-$w$ planes. This axial Killing field $\phi^a$ vanishes only at the origin $x=y=z=w=0$, so $\widetilde{M} = R \times (R^4 - \{ 0 \}) = R^2 \times S^3$. Thus, $\widetilde{M}$ is simply connected, whereas any manifold of the form ${\mathcal O} \times G$ with $G=R \times U(1)$ cannot be simply connected, so $\widetilde{M}$ cannot be of the form ${\mathcal O} \times G$.

We see no obvious reason why the action of $G$ for a stationary, axisymmetric black hole in $5$ or higher dimensions could not be as in the above example. Thus, our assumption of a trivial action of $G$ appears to be a genuine restriction. 

Let $(t, \varphi^\Lambda)$ be the natural coordinates on $G$---so  $-\infty < t < \infty$ and each $\varphi^\Lambda$ is periodic with period $2 \pi$---with group multiplication corresponding to addition. Let $x^\mu$ denote (local) coordinates on $\mathcal O$. Our assumption of a trivial action of $G$ on $M$ implies that on $\widetilde{M}$ we can {\em globally} choose coordinates $(t, \varphi^\Lambda)$ and (locally) choose coordinates $x^\mu$ so that on $\widetilde{M}$ the metric, $g_{ab}$, takes the form
\be
ds^2 = - \alpha dt^2 + 2 \beta_\Lambda dt d\varphi^\Lambda + \Phi_{\Lambda \Gamma} d\varphi^\Lambda  d\varphi^\Gamma + 2 B_\mu dt dx^\mu 
+ 2 A_{\Lambda \mu} d \varphi^\Lambda d x^\mu + \lambda_{\mu \nu} dx^\mu dx^\nu \, ,
\label{metricform}
\ee
where the metric components $(\alpha, \beta_\Lambda, \Phi_{\Lambda\Gamma}, B_\mu, A_{\Lambda \mu}, \lambda_{\mu \nu})$ do not depend upon $(t, \varphi^\Lambda)$. In these coordinates, the Killing fields are $t^a = (\partial/\partial t)^a$ and ${\phi_\Lambda}^a = (\partial/\partial \varphi^\Lambda)^a$. The coordinates $(t, \varphi^\Lambda)$ on $\widetilde{M}$ are unique up to $t \to t + f(x^\mu)$, $\varphi^\Lambda \to \varphi^\Lambda + f^\Lambda (x^\mu)$. The diffeomorphism defined by $(t, \varphi^\Lambda, x^\mu) \to (-t, -\varphi^\Lambda, x^\mu)$ will be an isometry of the metric \eqref{metricform} if and only $B_\mu = A_{\Lambda \mu} = 0$. Thus, one way of formulating our main goal is to show that we can set $B_\mu = A_{\Lambda \mu} = 0$ by making use of the coordinate freedom in $(t, \varphi^\Lambda)$.

The above coordinates break down at the ``axes,'' $\mathcal A$, where the Killing fields become linearly dependent. In Appendix \ref{axisappendix}, we show that under the assumption of a trivial action of $G$, the following {\em axis regularity conditions} hold: Coordinates $(t, \varphi^\Lambda)$ for the metric form \eqref{metricform} can be chosen on $\widetilde{M}$ so that (i) On any $t = {\rm const}$, $\varphi^\Lambda = {\rm const}$ surface, the metric $\lambda_{\mu \nu} dx^\mu dx^\nu$ smoothly extends to $\mathcal A$. (ii) The scalar fields $\Phi^{\Lambda \Gamma} \beta_\Gamma$ and one-form fields $\Phi^{\Lambda \Gamma} A_{\Gamma \mu} dx^\mu$ smoothly extend to $M$, where $\Phi^{\Lambda \Gamma}$ denotes the inverse of $\Phi_{\Lambda \Gamma} = {\phi_\Lambda}^a \phi_{\Gamma a}$. (iii) The one-form field $B_\mu dx^\mu$ and the tensor field $\Phi_{\Lambda \Gamma} d\varphi^\Lambda d\varphi^\Gamma + \lambda_{\mu \nu} dx^\mu dx^\nu$ smoothly extend to $M$.

In addition to our above assumption on the nature of the group action, we require that $(M, g_{ab})$ be an asymptotically flat, vacuum\footnote{For simplicity, we restrict consideration to vacuum general relativity in this paper, but as discussed in section \ref{gensec}, significant generalizations should be possible.} solution of Einstein's equation that contains a strongly asymptotically predictable black hole with a non-degenerate horizon. As is well known, the rigidity theorem \cite{MI, HIW} implies that the event horizon of a stationary black hole with a non-degenerate horizon must be a Killing horizon. Since the surface gravity, $\kappa$, of the event horizon is nonvanishing, the event horizon is (a portion of) a bifurcate Killing horizon. The rigidity theorem further states that the Killing field, $k^a$, that is normal to the horizon takes the form
\be
k^a = t^a + \sum_\Lambda \Omega^\Lambda {\phi_\Lambda}^a \, ,
\label{horkvf}
\ee
where each ${\phi_\Lambda}^a$ is an axial Killing field and these axial Killing fields mutually commute with each other and with $t^a$. The coefficient $\Omega^\Lambda$ defines the angular velocity of the horizon with respect to the Killing field ${\phi_\Lambda}^a$. We note that an arbitrary, smooth, asymptotically flat perturbation of $(M, g_{ab})$ satisfies the first law of black hole mechanics
\be
\delta M = \frac{1}{8 \pi} \kappa \delta A + \sum_\Lambda \Omega^\Lambda \delta J_\Lambda \, ,
\label{1stlaw}
\ee
where $M$ is the ADM mass, $A$ is the area of the event horizon, and $J_\Lambda$ are the ADM angular momenta. This relation will play a crucial role in our proof. An additional fact that we will need in our proof is that, by the topological censorship theorem \cite{Friedmanetal, ChruscielGalloway}, the domain of outer communications is simply connected.

The first key result that we shall use is the existence \cite{ChruscielWald} of an asymptotically flat Cauchy surface $\Sigma$ for the domain of outer communications that terminates at the bifurcation surface, $\mathcal B$, and is maximal in the sense that ${K^i}_i=0$, where $K_{ij}$ denotes the extrinsic curvature of $\Sigma$. Furthermore, $t^a$ is transverse to $\Sigma$, so by acting on $\Sigma$ with the time translation isometries, we obtain a foliation, $\Sigma_t$, of maximal Cauchy surfaces. Note that the results of \cite{ChruscielWald} assumed $d=4$ spacetime dimensions, but their proofs do not, in fact, depend upon the number of dimensions. 

We claim now that any axial Killing field ${\phi_\Lambda}^a$ must be tangent to each $\Sigma_t$. To prove this, let $n^a$ denote the unit normal to $\Sigma_t$, and let 
\be
N_\Lambda = - {\phi_\Lambda}^a n_a
\ee
be the ``lapse function'' associated with ${\phi_\Lambda}^a$. Since ${\phi_\Lambda}^a$ is a Killing field, the surface $\Sigma_t'$ obtained by displacing $\Sigma_t$ along ${\phi_\Lambda}^a$ must also be maximal. Therefore, $N_\Lambda$ must satisfy (see \cite{BrillFlaherty})
\be
-D^a D_a N_\Lambda +  K^{ab} K_{ab}N_\Lambda= 0
\ee
with the conditions $N_\Lambda \to 0$ at infinity (since ${\phi_\Lambda}^a$ must become asymptotically tangent to $\Sigma_t$) and $N_\Lambda=0$ at $\mathcal B$ (since $\mathcal B$ must be mapped into itself by any isometry, so ${\phi_\Lambda}^a$ must be tangent to $\mathcal B$). It follows immediately \cite{CantorBrill, ChruscielDelay} that $N_\Lambda=0$, as we desired to show.

Since $t^a$ is transverse to $\Sigma$ and all ${\phi_\Lambda}^a$ are tangent to $\Sigma_t$, it is clear that $t^a$ is linearly independent of ${\phi_\Lambda}^a$ in the domain of outer communications. Furthermore, we may use $t$ as a time coordinate in a coordinate system on $\widetilde{M}$ of the type \eqref{metricform}. This justifies the assumptions concerning $t^a$ and $t$ that are made in Appendix \ref{axisappendix}.

As explained in section \ref{intro}, our strategy is to find a map $i_\Sigma: \Sigma \to \Sigma$ such that $i_\Sigma^* (h_{ij}) = h_{ij}$, 
$i_\Sigma^* ({\phi_\Lambda}^i) = -{\phi_\Lambda}^i$, and $i_\Sigma^* (K_{ij}) = - K_{ij}$, where $h_{ij}$ is the induced metric on $\Sigma$, and $K_{ij}$ is the extrinsic curvature of $\Sigma$. It should be noted that, since we are now viewing the axial Killing fields as vector fields on $\Sigma$ rather than $M$, we have denoted them as ${\phi_\Lambda}^i$ rather than ${\phi_\Lambda}^a$. Given such an $i_\Sigma$, the desired ``$t$-$\phi$'' reflection isometry, $i$, is obtained by mapping a point $p$ lying at proper time $\tau$ along a normal geodesic starting at $s \in \Sigma$ to the point $q$ lying at proper time $-\tau$ along a normal geodesic starting at $i_\Sigma (s)$. Since $i$ maps the initial data, $(h_{ij}, K_{ij})$, on $\Sigma$ into itself, by uniqueness of Cauchy evolution, $i$ is an isometry.

The induced spatial metric, $h_{ij}$, on $\Sigma$ takes the form
\be
ds^2 = \Phi_{\Lambda \Gamma} d\varphi^\Lambda  d\varphi^\Gamma + 2 A_{\Lambda \mu} d \varphi^\Lambda d x^\mu + \lambda_{\mu \nu} dx^\mu dx^\nu \, .
\label{metform2}
\ee
It will possess a reflection isometry, $i_\Sigma$, of the form $(\varphi^\Lambda, x^\mu) \to (-\varphi^\Lambda, x^\mu)$ if and only if $A_{\Lambda \mu}$ can be set to zero after redefinitions of the form $\varphi^\Lambda \to \varphi^\Lambda + f^\Lambda(x^\mu)$. The resulting isometry $i_\Sigma$
will map ${\phi_\Lambda}^i \to -{\phi_\Lambda}^i$ and it will leave invariant every axisymmetric vector field that is orthogonal to all ${\phi_\Lambda}^i$. 

We can uniquely decompose any tensor field ${T^{ij\dots}}_{kl\dots}$ on $\Sigma$ into its parts parallel and perpendicular to ${\phi_\Lambda}^i$. For example, any vector field $W^i$, on $\Sigma$ can be uniquely written as
\be
W^i =  w^i + W^\Lambda {\phi_\Lambda}^i \, ,
\ee
with $w^i$ orthogonal to each ${\phi_\Lambda}^{i}$. Note that the quantities $W^\Lambda$ are scalar fields on $\Sigma$. The extrinsic curvature can be decomposed as 
\be \label{Kdecomp}
K_{ij} = k_{ij} + K^{\Lambda\Gamma}  \phi_{\Lambda i} \phi_{\Gamma j}  + 2  {\mathcal K^\Lambda}_{(i}\phi_{|\Lambda| j)}  \, ,
\ee
where $k_{ij} {\phi_\Lambda}^i=0$ and ${\mathcal K^\Lambda}_i {\phi_\Gamma}^i=0$.
For any axisymmetric tensor field ${T^{ij\dots}}_{kl\dots}$ on $\Sigma$ we define its {\em polar part} to be the terms in such a decomposition that contain an even number of ${\phi_\Lambda}^i$, and we define its {\em axial part} to be the terms that contain an odd number of ${\phi_\Lambda}^i$. Thus, for an axisymmetric vector field $W^i$, we call $w^i$ its ``polar part'' and $W^\Lambda {\phi_\Lambda}^{i}$ its ``axial part''. For the extrinsic curvature, $K_{ij}$, the first two terms in \eqref{Kdecomp} are its polar part whereas the last term is its axial part. The extrinsic curvature will reverse sign under $i_\Sigma$ if and only if its polar part vanishes.

Thus, the existence of a $(t$-$\phi)$-reflection isometry will be proven if we can show that $A_{\Lambda \mu}$ can be set to zero and that the polar part of $K_{ij}$ vanishes. Before proceeding to show this, it is useful to reformulate this question more geometrically, using the notion of fiber bundles.

\section{Fiber Bundle Structure} \label{manifoldoforbits}

As shown in the previous section, the axial Killing fields are tangent to the maximal hypersurface $\Sigma$, and the Killing field $t^a$ is transverse to $\Sigma$. Thus, when restricted to $\Sigma$, the spacetime isometry group $G$ reduces to 
\be
G_\Sigma = [U(1)]^n \, .
\ee
The trivial action of $G$ on $M$ implies that $G_\Sigma$ has a correspondingly trivial action on $\Sigma$. Specifically, $\widetilde{\Sigma} \equiv \widetilde{M} \cap \Sigma$ is an open, dense subset of $\Sigma$ that has the structure $\widetilde{\Sigma} = {\mathcal O} \times G_\Sigma$, where $\mathcal O$ is the manifold of orbits of $G$ (see the first paragraph of section \ref{assumptions} above). This trivial action of $G_\Sigma$ gives $\widetilde{\Sigma}$ the structure of a trivial principal fiber bundle with fiber group $G_\Sigma$ and base space $\mathcal O$. Although there is no necessity for introducing fiber bundle machinery to describe structures on a cross-product space, it is very convenient to do so in order to describe in an invariant, geometrical manner the various quantities that arise. It is also useful to do so in order to draw parallels with calculations in Yang-Mills theory, and allow generalizations (not undertaken here) to nontrivial actions of $G_\Sigma$. The constructions below may be viewed as a generalization of \cite{Geroch, Geroch2} to the case of $n$ commuting axial Killing vector fields, but specialized to the Riemannian signature of $h_{ij}$. 

Let $\mathcal V$ denote the Lie algebra of $G_\Sigma$. Then $\mathcal V$ is an $n$-dimensional vector space with trivial Lie bracket. We denote elements of $\mathcal V$ using an upper case Latin upper index, e.g., $v^A \in \mathcal V$. There is a natural, one-to-one correspondence between $\mathcal V$ and Killing fields on $\widetilde{\Sigma}$. Consequently, the tangent space to the ``fibers'' (i.e., the orbits of $G_\Sigma$) at each point $q \in \widetilde{\Sigma}$ are thereby naturally isomorphic to $\mathcal V$. Our original collection of axial Killing fields ${\phi_\Lambda}^i$, $\Lambda = 1, \dots, n$, may therefore be viewed as being associated with a particular basis, ${v_\Lambda}^A$, $\Lambda = 1, \dots, n$, of the Lie algebra $\mathcal V$. It is therefore natural to view the axial Killing fields as comprising a single object, ${\phi_A}^i$, i.e., a vector field on $\widetilde{\Sigma}$ that is valued in the dual space to the Lie algebra $\mathcal V$. The Killing field ${\phi_\Lambda}^i$ may then be recovered by contraction with the basis vector ${v_\Lambda}^A$, i.e.,
\be
{\phi_\Lambda}^i = {v_\Lambda}^A {\phi_A}^i \, .
\ee

We will be interested in the tensor fields on $\widetilde{\Sigma}$ that are axisymmetric in the sense of having vanishing Lie derivative with respect to ${\phi_A}^i$; such tensor fields are called {\em equivariant} in the fiber bundle terminology. Clearly, any tensor field on $\widetilde{\Sigma}$ that is constructed from $h_{ij}$ and ${\phi_A}^i$ will be axisymmetric. Any axisymmetric tensor field ${T^{ij\dots}}_{kl\dots}$ on $\widetilde{\Sigma}$ that is orthogonal to ${\phi_A}^i$ in all of its (spatial) indices can be projected to the base space $\mathcal O$. Our strategy is to do this projection and formulate all relations as relations holding on $\mathcal O$. We will not make any notational distinction between an axisymmetric tensor field ${T^{ij\dots}}_{kl\dots}$ on $\widetilde{\Sigma}$ that is orthogonal to ${\phi_A}^i$ and its projection to $\mathcal O$, i.e., we will also denote its projection to $\mathcal O$ by ${T^{ij\dots}}_{kl\dots}$.

Two important examples of tensor fields on $\mathcal O$ that are obtained by projection in this manner are as follows. First, let
\be
\Phi_{AB}\equiv  h_{ij} {\phi_A}^i {\phi_B}^j \, .
\ee
Clearly $\Phi_{AB}$ is axisymmetric and, since it is a scalar quantity (valued in the space of tensors of type $(0,2)$ over $\mathcal V$), it is well defined on $\mathcal O$. Indeed, since $h_{ij}$ is positive definite, it follows that $\Phi_{AB}$ yields a positive definite metric on $\mathcal V$ (which, of course, depends upon the point $p \in \mathcal O$). The components, $\Phi_{\Lambda \Sigma}$, of $\Phi_{AB}$ have already appeared in the metric form \eqref{metform2}.

Second, let
\be \label{metric}
\mu_{ij} = h_{ij} - \Phi^{AB} \phi_{Ai} \phi_{Bj} \, ,
\ee
where $\Phi^{AB}$ is the inverse metric of $\Phi_{AB}$ and we have lowered the spatial index of ${\phi_A}^i$ using $h_{ij}$. It follows immediately that $\mu_{ij} = \mu_{ji}$ and that $\mu_{ij}$ is axisymmetric. It also is easily verified that $\mu_{ij} {\phi_A}^i = 0$, so $\mu_{ij}$ projects to $\mathcal O$, where it defines a Riemannian metric on $\mathcal O$. The components, $\mu_{\mu \nu}$, of $\mu_{ij}$ are related to the quantities appearing in \eqref{metform2} by
\be
\mu_{\mu \nu} = \lambda_{\mu \nu} - \Phi^{\Lambda \Gamma} A_{\Lambda \mu} A_{\Gamma \nu} \, .
\ee

Now, consider the quantity
\be
{\phi^A}_i = \Phi^{AB} h_{ij} {\phi_B}^j
\ee
obtained from the Killing fields ${\phi_A}^i$ by raising the Lie algebra index with $\Phi^{AB}$ and lowering the spatial index with $h_{ij}$. Then, obviously, ${\phi^A}_i$ is a Lie-algebra-valued one-form on $\widetilde{\Sigma}$. Furthermore, this one-form satisfies,
\be
{\phi^A}_i {\phi_B}^i = \Phi^{AC} h_{ij} {\phi_C}^j {\phi_B}^i = \Phi^{AC} \Phi_{CB} = {\delta^A}_B \, .
\label{phiprop}
\ee
It follows that ${\phi^A}_i$ maps an arbitrary ``vertical vector'' $v^i$---i.e., a vector $v^i$ that is tangent to the fibers and thus can be written as a linear combination of ${\phi_\Lambda}^i$---into the corresponding vector $v^A$ in the Lie algebra $\mathcal V$. By definition, a connection is an equivariant Lie algebra valued one-form with this property. {\em Thus, ${\phi^A}_i$ defines a connection on the principal fiber bundle $\widetilde{\Sigma}$.} 

We now use the assumed trivial action of $G_\Sigma$ to encode the information contained in ${\phi^A}_i$ as a field ${A^A}_i$ on $\mathcal O$. Choose a cross-section, $\mathcal S$, of $\widetilde{\Sigma}$. Then all of the nontrivial information contained in ${\phi^A}_i$ is contained in the pullback of ${\phi^A}_i$ to $\mathcal S$. Let ${A^A}_i$ denote the equivariant (i.e., axisymmetric) one-form field on $\widetilde{\Sigma}$ whose pullback to $\mathcal S$ agrees with the pullback of ${\phi^A}_i$ to $\mathcal S$ and is such that ${A^A}_i {\phi_B}^i = 0$. The quantities $A_{\Lambda \mu}$ appearing in \eqref{metform2} are just the components of $A_{Ai} = \Phi_{AB} {A^B}_i$. As shown in Appendix \ref{axisappendix}, a cross-section $\mathcal S$ can be chosen so that ${A^A}_i$ smoothly extends to $\Sigma$, and we assume that such a choice of cross-section has been made. Finally, since ${A^A}_i {\phi_B}^i = 0$, we may view ${A^A}_i$ as a $\mathcal V$-valued one-form field on $\mathcal O$. However, unlike $\Phi_{AB}$ and $\mu_{ij}$, the quantity ${A^A}_i$ is ``gauge dependent'' in that it depends upon a choice of cross-section $\mathcal S$.

From the above, we see that all of the information contained in the coordinate components \eqref{metform2} of an axisymmetric Riemannian metric, $h_{ij}$, on $\widetilde{\Sigma}$ is encoded in the following 3 objects: (1) a scalar field $\Phi_{AB}$ on $\mathcal O$ valued in Riemannian metrics on $\mathcal V$; (2) a Riemannian metric $\mu_{ij}$ on $\mathcal O$; and (3) a $\mathcal V$-valued one-form field ${A^A}_i$ on $\mathcal O$. Conversely, if one specifies the Killing fields ${\phi_A}^i$ on $\widetilde{\Sigma}$ and the cross-section $\mathcal S$, then any choice of a field $\Phi_{AB}$ on $\mathcal O$ that is valued in Riemannian metrics on $\mathcal V$, a Riemannian metric $\mu_{ij}$ on $\mathcal O$, and a one-form field ${A^A}_i$ on $\mathcal O$ uniquely determines an axisymmetric Riemannian metric $h_{ij}$ on $\widetilde{\Sigma}$.

The \emph{the curvature 2-form of ${\phi^A}_i$} is defined by
\be
{F^A}_{ij} = D_i {\phi^A}_j - D_j {\phi^A}_i 
\ee
where $D_i$ denotes the metric compatible derivative operator on $\widetilde{\Sigma}$. Of course, the antisymmetrized derivative of a differential form is independent of the choice of derivative operator, so we may write the definition of ${F^A}_{ij}$ in differential forms notation as
\be \label{F1}
\boldsymbol{F}^A = d \boldsymbol{\phi}^A \, .
\ee
As is well known, the curvature $\boldsymbol{F}^A$ is ``horizontal," i.e. it is orthogonal to the Killing fields ${\phi_A}^i$, and thus may be projected to $\mathcal O$. To see this, we note the the general formula for the Lie derivative of a differential form yields
\be
\iota_{\phi_A} \boldsymbol{F}^B = \pounds_{\phi_A} \boldsymbol\phi^B - d (\iota_{\phi_A} \boldsymbol\phi^B) \, ,
\ee
where $\iota_{\phi_A}$ denotes the contraction of ${\phi_A}^i$ into the first index of a differential form. The first term on the right side vanishes by axisymmetry, and the second term vanishes because ${\phi_A}^i {\phi^B}_i = {\delta^B}_A$, which has vanishing derivative. Thus, ${\phi_A}^i {F^B}_{ij} = 0$, as we desired to show. Finally, we note that we also have
\be
\boldsymbol{F}^A = d \boldsymbol{A}^A \, .
\label{F2}
\ee
To show this, we note that the pullbacks of both sides to $\mathcal S$ agree, since ``$d$'' commutes with pullbacks. However, we have just shown that the left side is orthogonal to ${\phi_A}^i$ and a similar calculation shows that the right side is also orthogonal to ${\phi_A}^i$. Thus, \eqref{F2} holds. Note that since, by the axis regularity conditions, ${A^A}_i$ can be smoothly extended to $\Sigma$, it follows that ${F^A}_{ij}$ also can be smoothly extended to $\Sigma$ and \eqref{F2} holds everywhere on $\Sigma$.

The Frobenius condition for (local) surface orthogonality of the span of the Killing fields $\{ {\phi_\Lambda}^i \}$ (i.e., surface orthogonality of the fibers) is that there exist one-forms ${w^A}_{Bi}$ such that
\be \label{dphi2}
d \boldsymbol \phi_A = \boldsymbol\phi_B \wedge {\boldsymbol w^B}_{A} \, . 
\ee
On the other hand, we have
\be \label{dphi}
d \boldsymbol \phi_A = d (\Phi_{AB} \boldsymbol \phi^B) = \Phi_{AB} \boldsymbol{F}^B +  (d\Phi_{AB}) \wedge \boldsymbol \phi^B \, .
\ee
Since $\boldsymbol{F}^B$ is horizontal, we see that the necessary and sufficient condition for (local) surface orthogonality is simply $\boldsymbol{F}^B = 0$.

It is useful to define the derivative operator, $\mathcal D_i$, on $\mathcal O$ associated with $\mu_{ij}$ as follows: If ${T^{i_1 i_2 \dots}}_{j_1 j_2 \dots}$ is a tensor field on $\widetilde\Sigma$ that projects to $\mathcal O$, define
\be
\mathcal D_k {T^{i_1 i_2 \dots}}_{j_1 j_2 \dots} = \left({\mu^{i_1}}_{l_1}{\mu^{i_2}}_{l_2}\cdots\right)  \left({\mu_{j_1}}^{m_1}{\mu_{j_2}}^{m_2}\cdots \right) {\mu_{k}}^{n}D_n {T^{l_1 l_2 \dots}}_{m_1 m_2 \dots} \, .
\ee
This relationship between the derivative operator of $\mu_{ij}$ and the derivative operator of $h_{ij}$ can be straightforwardly used to calculate the Riemann curvature of $\mu_{ij}$, denoted $\mathcal R_{ijkl}$, in terms of the curvature of $h_{ij}$, denoted\footnote{Normally, one would use ${}^{(d-1)}R_{ijkl}$ to denote the Riemann curvature of $h_{ij}$, and reserve the symbol $R_{abcd}$ for the Riemann curvature of the spacetime metric, $g_{ab}$. However we will not have occasion here to use the curvature of $g_{ab}$, so we will drop the $(d-1)$ superscript on the curvature of $h_{ab}$. } $R_{ijkl}$, and derivatives of the axial Killing vector fields.
This calculation is performed in \cite{Geroch}, and the result is
\be
\mathcal R_{i j k l} = {\mu_{[i}}^{p} {\mu_{j]}}^{q} {\mu_{[k}}^{r} {\mu_{l]}}^{s} \biggl( R_{pqrs} + 2 \Phi^{AB}\Bigl[ (D_{p}\phi_{Aq})(D_{r}\phi_{Bs})+(D_{p}\phi_{Ar})(D_{q}\phi_{Bs})\Bigr] \biggr) \, .
\ee  
By performing the necessary contractions, one finds the scalar curvature of $\mu_{ij}$ to be given by
\be \label{mathcalR}
\mathcal R = \mu^{ik} \mu^{jl} R_{ijkl} + 3 \Phi^{AB} \mu^{ij} \mu^{kl}(D_i\phi_{Ak})(D_j\phi_{Bl}) \, .
\ee
This equation will be useful in writing the constraint equations as equations on $\mathcal O$, which we do in the next section.

\section{The Constraints} \label{constraintssec}

In this section, we rewrite the Einstein constraint equations on $\Sigma$ for axisymmetric initial data $(h_{ij}, K_{ij})$, as equations on the manifold of orbits $\mathcal O$ in terms of $\Phi_{AB}$, $\mu_{ij}$, ${F^A}_{ij}$, and the axial and polar parts of the extrinsic curvature. Since ${K^i}_i = 0$ on $\Sigma$, the constraint equations are
\begin{align}
&D^i K_{ij} = 0 \label{mconstraint}\\
&R-K^{ij}K_{ij}=0 \label{Hconstraint} \, .
\end{align}
We will now write these equations entirely in terms of the above tensor fields on $\mathcal O$. To do so, we will make frequent use of \eqref{dphi}, which, since ${\phi_A}^i$ are Killing fields, can be written in the form
\be \label{Dphi}
D_i \phi_{Aj} = \frac{1}{2} \Phi_{AB}{F^B}_{ij} - \Phi^{BC}\phi_{C[i}\mathcal D_{j]} \Phi_{AB} \, .
\ee
Contracting both sides with ${\phi_D}^i$ yields another useful identity,
\be \label{DPhi}
{\phi_D}^i D_i \phi_{Aj} = -\frac{1}{2} \mathcal D_j \Phi_{AD} \, .
\ee

First, we perform the polar-axial decomposition of $K_{ij}$ via \eqref{Kdecomp}. Since $K_{ij}$ is axisymmetric, it immediately follows that $k_{ij}$, $K^{AB}$, and ${\mathcal K^A}_i$ can be projected to tensor fields on $\mathcal O$ (of symmetric tensor, scalar, and dual-vector type, respectively).

The momentum constraint, \eqref{mconstraint}, can be decomposed by projecting its free index either along ${\phi_A}^i$ or orthogonal to ${\phi_A}^i$.
Doing the former, we obtain as the ``axial part'' of the momentum constraint
\ba
0 &= {\phi_A}^{j} D^i K_{ij} =  D^i \left({\phi_A}^{j}K_{ij}\right) = D^i \mathcal K_{A i} \, , 
\ea
where Killing's equation was used in the second equality, and the decomposition \eqref{Kdecomp} together with the axisymmetry of ${K_A}^B$ was used in the third equality. Writing $D^i \mathcal K_{A i} = h^{ij} D_j  {\mathcal K}_{A i}$, substituting for $h^{ij}$ from \eqref{metric}, and using the definition of $\mathcal D_i$, the axial part of the momentum constraint becomes
\be
0 = \mathcal D^i \mathcal K_{Ai} + \Phi^{BC} {\phi_B}^{i} {\phi_C}^{j}D_j \mathcal K_{Ai} = \mathcal D^i \mathcal K_{Ai} - \Phi^{BC}\mathcal K_{Ai}   {\phi_C}^{j}D_j{\phi_B}^{i} \, .
\ee
On the other hand, from \eqref{DPhi} we have
\be
\Phi^{BC} {\phi_C}^j D_j {\phi_B}^i = - \frac{1}{2} \Phi^{BC} \mathcal D^i \Phi_{BC} = - \frac{1}{\sqrt{\Phi}} {\mathcal D}^i \sqrt{\Phi} \, ,
\ee
where $\Phi$ denotes the determinant\footnote{To define the determinant of $\Phi_{AB}$, we must choose an arbitrary fixed (i.e., $x$-independent) nonvanishing antisymmetric tensor $\epsilon^{A_1 \dots A_n}$ on the Lie algebra $\mathcal V$. A different choice of $\epsilon^{A_1 \dots A_n}$ will rescale the determinant by a constant factor, which will cancel out in all of our formulas.} of $\Phi_{AB}$. Putting this all together, we find that the axial part of the momentum constraint is
\be \label{momentumconstraint1}
\frac{1}{\sqrt{\Phi}}\mathcal D^i \left(\sqrt{\Phi}\, \mathcal K_{Ai}\right) = 0 \, .
\ee

The polar part of the momentum constraint is
\be \label{polarconstraint1}
0 = {\mu_j}^k D^i K_{ik} =  {\mu_j}^k D^i k_{ik} +{\mu_j}^k K^{AB}  \phi_{Ai} D^i   \phi_{Bk}
+{\mu_j}^k  {\mathcal K^A}_i D^i \phi_{Ak}+{\mu_j}^k    \phi_{Ai} D^i  {\mathcal K^A}_k  \, .
\ee
Performing manipulations similar to those used in the derivation of \eqref{momentumconstraint1} above, we find that the first term on the right side takes the form
\be\label{polarconstraint0}
{\mu_j}^k D^i k_{ik} = \frac{1}{\sqrt{\Phi}}\mathcal D^i \left(\sqrt{\Phi}\, k_{ij}\right) \, .
\ee
Using the axisymmetry of ${\mathcal K^A}_i$, the fourth term on the right side of \eqref{polarconstraint1} can be written as
\be
\mu_{jk}    {\phi_A}^i D_i  \mathcal K^{Ak} = \mu_{jk} \mathcal K^{Ai}    D_i  {\phi_A}^k  \, .
\ee
We may now use \eqref{Dphi} to eliminate the derivatives of ${\phi_A}^i$ in the last three terms of \eqref{polarconstraint1}, which can then be combined with \eqref{polarconstraint0} to yield
\be \label{momentumconstraint2}
\frac{1}{\sqrt{\Phi}}\mathcal D^i \left(\sqrt{\Phi}\, k_{ij}\right) -\frac{1}{2} K^{AB} \mathcal D_j \Phi_{AB} + {\mathcal K_A}^{i}  {F^A}_{ij} = 0 \, .
\ee

Finally, we turn our attention to the Hamiltonian constraint, \eqref{Hconstraint}.
The scalar curvature of $h_{ij}$ can be written
\be \label{scalarR1}
R = h^{ik} h^{jl} R_{ijkl} = \mu^{ik}\mu^{jl}R_{ijkl} + \left(2   \mu^{jl} + \Phi^{CD}  {\phi_C}^{j} {\phi_D}^{l} \right)\Phi^{AB}{\phi_A}^{i} {\phi_B}^{k}R_{ijkl} \, .
\ee
By \eqref{mathcalR} and \eqref{Dphi}, the first term on the right side is
\be \label{R1}
\mu^{ik}\mu^{jl} R_{ijkl} = \mathcal R - 3 \Phi^{AB} \mu^{ij} \mu^{kl}(D_i\phi_{Ak})(D_j\phi_{Bl}) = \mathcal R - \frac{3}{4}  {F^A}_{ij}    {F_A}^{ij} \, . 
\ee
To calculate the remaining terms in \eqref{scalarR1}, we eliminate $R_{ijkl}$ using the relation
\be
D_j D_k \phi_{Al}={\phi_A}^{i}R_{ijkl} \, ,
\ee
which is satisfied by any Killing vector field. We then have
\ba \label{DDphi}
{\phi_B}^{k}{\phi_A}^{i}R_{ijkl}={\phi_B}^{k}D_j D_k \phi_{Al} &= D_j\left({\phi_B}^{k} D_k \phi_{Al}\right) - \left(D_j{\phi_B}^{k} \right)\left( D_k \phi_{Al} \right)\\
&= -\frac{1}{2}D_j D_l \Phi_{AB} - \left(D_j{\phi_B}^{k} \right)\left( D_k \phi_{Al} \right) \, ,
\ea
where \eqref{DPhi} was used in the second line. Contracting \eqref{DDphi} with $2\mu^{jl}\Phi^{AB}$ and eliminating derivatives of ${\phi_A}^i$ with \eqref{Dphi}, we obtain
\be \label{R2}
2\mu^{jl}\Phi^{AB}{\phi_A}^{i} {\phi_B}^{k}R_{ijkl} = -\Phi^{AB}\mathcal D^2 \Phi_{AB} + \frac{1}{2} {F^A}_{ij}    {F_A}^{ij} + \frac{1}{2} \Phi^{AC}\Phi^{BD}(\mathcal D_i\Phi_{AB})(\mathcal D^{i}\Phi_{CD}) \, .
\ee
On the other hand, contracting \eqref{DDphi} with $\Phi^{CD}  {\phi_C}^{j} {\phi_D}^{l}$, we obtain
\be \label{R3}
\Phi^{CD}  {\phi_C}^{j} {\phi_D}^{l} \Phi^{AB} {\phi_A}^{i} {\phi_B}^{k}R_{ijkl} =\frac{1}{4}  \left(\Phi^{AC}\Phi^{BD}  - \Phi^{AB} \Phi^{CD} \right) \left(\mathcal D_i \Phi_{AB} \right)\left(\mathcal D^i \Phi_{CD}\right) \, .
\ee
Putting \eqref{R1}, \eqref{R2}, and \eqref{R3} together, we obtain our final result for the scalar curvature,
\be
R = \mathcal R - \frac{1}{4} {F^A}_{ij}    {F_A}^{ij} -\Phi^{AB}\mathcal D^2 \Phi_{AB} + \frac{1}{4}\left(3\Phi^{AC}\Phi^{BD}  - \Phi^{AB} \Phi^{CD} \right) \left(\mathcal D_i \Phi_{AB} \right)\left(\mathcal D^i \Phi_{CD}\right) \, .
\ee
Since we have
\be
K^{ij}K_{ij} = k^{ij}k_{ij} + K^{AB}K_{AB} + 2 \mathcal K^{Ai} \mathcal K_{Ai} \, ,
\ee
the Hamiltonian constraint equation, \eqref{Hconstraint}, takes the form
\ba \label{hamiltonianconstraint}
\mathcal R  -\Phi^{AB}\mathcal D^2 \Phi_{AB} &+ \frac{1}{4}\left(3\Phi^{AC}\Phi^{BD}  - \Phi^{AB} \Phi^{CD} \right) \left(\mathcal D_i \Phi_{AB} \right)\left(\mathcal D^i \Phi_{CD}\right)  \\
&- 2 \mathcal K^{Ai} \mathcal K_{Ai} 
 -\left[ \frac{1}{4} {F^A}_{ij}    {F_A}^{ij} + k^{ij}k_{ij} + K^{AB}K_{AB} \right] = 0 \, .
\ea

Equations \eqref{momentumconstraint1}, \eqref{momentumconstraint2}, and \eqref{hamiltonianconstraint} express the constraint equations entirely in terms of the desired variables, which are well defined on $\mathcal O$. These equations therefore may be interpreted equally well as equations holding on the manifold $\widetilde{\Sigma}$ (with indices raised and lowered with $h^{ij}$ and $h_{ij}$) or as equations holding on the manifold $\mathcal O$ (with indices raised and lowered with $\mu^{ij}$ and $\mu_{ij}$). For our purposes, it will be most useful to view them as equations holding on $\mathcal O$.

\section{Construction of a Conformally Modified Scaling Perturbation} \label{conformalsection}

In this section, we will construct a linearized perturbation off of a stationary and axisymmetric black hole spacetime that will be used in the proof of the theorem of the next section. The perturbation will be the sum of a perturbation that scales up ${F^A}_{ij}$ and the polar parts of $K_{ij}$ plus a perturbation that conformally modifies $h_{ij}$ and $K_{ij}$. We will also establish some key properties of this perturbation.

Let $(h_{ij}, K_{ij})$ be initial data on a maximal slice $\Sigma$ for a stationary and axisymmetric, asymptotically flat, black hole solution to the vacuum Einstein equation, with trivial action of the stationary-axisymmetric symmetry group $G$. Let ${A^A}_i$ be as defined in the paragraph following \eqref{phiprop}. Consider the perturbation $(\delta_1h_{ij}, \delta_1K_{ij})$ to the initial data on $\Sigma$ given by
\be
\delta_1 h_{ij} = 2 \phi_{A(i} {A^A}_{j)} \, ; \qquad \delta_1 {\mathcal K^A}_i = 0 \, , \quad\delta_1 k_{ij} = k_{ij}\, , \quad \delta_1 K^{AB} = K^{AB} \,  ,
\label{pert1}
\ee
where ${\mathcal K^A}_i $, $k_{ij}$, and $K^{AB}$ were defined by the decomposition \eqref{Kdecomp}. Note that in addition to the contributions from $\delta_1 k_{ij}$ and $\delta_1 K^{AB}$, 
$\delta_1 K_{ij}$ will get contributions from $\delta_1 \phi_{Ai} = \delta_1 (h_{ij} {\phi_A}^j) = (\delta_1 h_{ij}) {\phi_A}^j = A_{Ai}$ in \eqref{Kdecomp}.
In terms of the variables introduced in section \ref{manifoldoforbits}, the metric perturbation defined by \eqref{pert1} satisfies 
\be
\delta_1 \Phi_{AB} = \delta_1 \mu_{ij} = 0 \, , \qquad   \delta_1 {A^A}_i = {A^A}_i \, .
\ee
Note that both $\delta_1 h_{ij}$ and $\delta_1 K_{ij}$ extend smoothly to the axes $\mathcal A$, i.e., they define smooth perturbations on $\Sigma$.

We have chosen the perturbation \eqref{pert1} so that $\delta_1 h_{ij}$ is purely axial and $\delta_1 K_{ij}$ is purely polar (see the end of section \ref{assumptions}). In terms of the variables appearing in the constraint equations given in the previous section (viewed as equations holding on tensors on the manifold of orbits), this perturbation scales up ${F^A}_{ij}$ and the polar parts, $k_{ij}$  and $K^{AB}$, of the extrinsic curvature, but it leaves the quantities $\Phi_{AB}$, $\mu_{ij}$, $\mathcal D_i$, and ${\mathcal K^A}_i$ unchanged. 
By inspection, the perturbation \eqref{pert1} satisfies the linearization of the momentum constraints \eqref{momentumconstraint1} and \eqref{momentumconstraint2}. However, it fails to satisfy the linearization of the Hamiltonian constraint \eqref{hamiltonianconstraint}. Indeed, we obtain
\be
\delta_1 \left(R-K^{ij}K_{ij}\right) =  -\frac{1}{2} {F^A}_{ij} {F_A}^{ij} - 2k^{ij}k_{ij} - 2K^{AB}K_{AB} \, .
\ee

To rectify this, we add the linearized conformal transformation
\ba \label{conformal}
\delta_2 h_{ij} &= \psi h_{ij}\\
\delta_2 K^{ij} &=- \left(\frac{d+1}{2}\right)\psi K^{ij} \, .
\ea
Since ${K^i}_i = 0$, this perturbation can be seen to satisfy the linearized momentum constraint for any choice of $\psi$ by the following calculation:
\ba \label{conformalmomentumconstraint}
\delta_2 \left( D_i K^{ij} \right) &= D_i \delta_2 K^{ij} + K^{k(i}h^{j)l} \left(D_i \delta_2 h_{kl} + D_k \delta_2 h_{il} - D_l \delta_2 h_{ik} \right) \\
&=- \left(\frac{d+1}{2}\right) K^{ij} D_i\psi + K^{k(i}h^{j)l} \left(h_{kl}D_i \psi   + h_{il}D_k \psi  - h_{ik} D_l \psi  \right) \\
&=- \left(\frac{d+1}{2}\right) K^{ij} D_i\psi + K^{ij}D_i \psi   + \frac{1}{2} (d-1)K^{ij}D_i \psi \\
&=0   \, .
\ea
On the other hand, the linearized conformal perturbation fails to satisfy the linearized Hamiltonian constraint. We have
\ba \label{conformalhamiltonianconstraint}
\delta_2 \left(R-K^{ij}K_{ij}\right) &= -h^{ij}D^2 \delta_2 h_{ij}+D^i D^j \delta_2 h_{ij}-R^{ij}\delta_2 h_{ij} - 2 K_{ij} \delta_2 K^{ij} - 2K^{ij}{K_j}^{k}\delta_2 h_{ik} \\
&= -(d-1)D^2 \psi + D^2 \psi - R \psi + (d+1)K^{ij}K_{ij}\psi - 2 K^{ij} K_{ij} \psi \\
&= (d-2)\left(-D^2\psi  +  K^{ij}K_{ij}\psi\right) \, ,
\ea
where the background constraint, $R=K_{ij}K^{ij}$, was used.

Now define $\mathcal Q$ by
\be
\mathcal Q \equiv \frac{1}{2} {F^A}_{ij}    {F_A}^{ij} + 2 k^{ij}k_{ij} +2 K^{AB}K_{AB}  \, .
\label{Q}
\ee
Clearly, we have $\mathcal Q \geq 0$. Let $\psi$ to be the solution to
\be \label{e2}
-D^2\psi  +  K^{ij}K_{ij}\psi = \frac{\mathcal Q}{d-2}
\ee
with boundary conditions $\psi = 0$ at the bifurcation surface $\mathcal B$ and $\psi \to 0$ at infinity. By standard arguments \cite{CantorBrill, ChruscielDelay} there exists a unique solution to \eqref{e2} satisfying these boundary conditions. It follows immediately that the perturbation
\be
\delta h_{ij} = \delta_1 h_{ij} + \delta_2 h_{ij} \, , \quad \quad \delta K_{ij} = \delta_1 K_{ij} + \delta_2 K_{ij}
\label{pert}
\ee
satisfies the linearized momentum and Hamiltonian constraints. This is the perturbation of interest, which we will use in the proof of the theorem of the next section. We have the following lemma on some key properties of this perturbation.

\begin{lemma}
The perturbation \eqref{pert} (see \eqref{pert1}, \eqref{conformal}, and \eqref{e2}) satisfies $\delta A = 0$ and $\delta J_A = 0$, where $A$ denotes the area of the bifurcation surface $\mathcal B$, and $J_A$ is the ADM angular momentum. Furthermore, if $\mathcal Q \neq 0$, then $\delta M > 0$, where $M$ is the ADM mass.
\end{lemma}

\begin{proof}
Since $\delta_1 h_{ij}$ is purely axial, we have $\delta_1 A = 0$. Since $\psi |_\mathcal B = 0$, we have $\delta_2 A = 0$. Hence, we have $\delta A = 0$.

The ADM angular momentum $J_A$ is given by
\be
J_A = \frac{1}{8\pi}\int_{S_\infty} K_{ij} {\phi_A}^ir^j \, ,
\ee
where $r^j$ denotes the outward pointing normal of a sphere, $S_\infty$, with radius taken to infinity. Only the axial part of extrinsic curvature contributes to this integral, so we may write (see \eqref{Kdecomp})
\be
J_A = \frac{1}{8\pi}\int_{S_\infty} 2 {\mathcal K^B}_{(i} \phi_{|B| j)} {\phi_A}^i r^j = \frac{1}{8\pi}\int_{S_\infty} {\mathcal K^B}_j \Phi_{AB} r^j = \frac{1}{8\pi}\int_{S_\infty} {\mathcal K}_{Aj} r^j \, .
\ee
We have $\delta_1 {\mathcal K}_{Aj} = 0$, and $\delta_1 r^j=-r_i A^{Bi}{\phi_B}^j$ is orthogonal to ${\mathcal K}_{Aj}$, so $\delta_1 J_A = 0$. On the other hand, since $\psi \to 0$ at infinity, we have $\delta_2 J_A = 0$. Hence, we obtain $\delta J_A = 0$, as we desired to show.

The perturbed ADM mass is given by
\be
\delta M = \frac{1}{16\pi}\int_{S_\infty} r^i h^{jk}\left( D_k \delta h_{ij} - D_i \delta h_{jk} \right) \, .
\label{admm}
\ee
Since $\delta_1 h_{ij}$ is purely axial, we have $\delta_1 M = 0$. On the other hand, substitution of $\delta_2 h_{ij} = \psi h_{ij}$ in \eqref{admm} yields
\be
\delta_2 M = -\frac{1}{16\pi}(d-2)\int_{S_\infty} r^i D_i \psi \, .
\ee
To evaluate this, let $\xi$ be the solution to
\be \label{homogeneous}
-D^2 \xi + K^{ab} K_{ab} \xi = 0 \, ,
\ee
with boundary conditions $\xi=0$ at $\mathcal B$ and $\xi\rightarrow 1$ at infinity; a unique solution exists by standard arguments\footnote{Note that by choosing a smooth function $f$ such that $f = 1$ in a neighborhood of infinity and $f = 0$ at $B$, the problem at hand can be converted to that of solving the inhomogeneous equation $-D^2 \chi + K^{ab} K_{ab} \chi =  F$ with standard boundary conditions $\chi = 0$ at $B$ and $\chi \to 0$ at infinity, where $F = D^2 f - K^{ab} K_{ab} f$.} \cite{CantorBrill, ChruscielDelay}. By the strong maximum principle (see, e.g., \cite{Gilbarg}), we have $\xi > 0$ everywhere on $\Sigma \setminus \mathcal B$. By \eqref{e2} and \eqref{homogeneous}, we have
\be
D^i\left( \xi D_i \psi - \psi D_i \xi \right) =  -\frac{ \xi\mathcal Q}{d-2} \, .
\ee
Integrating this equation over $\Sigma$ using the boundary conditions imposed on $\psi$ and $\xi$ at spatial infinity and at $\mathcal B$ yields
\be
\int_{S_\infty} r^i  D_i \psi = - \int_\Sigma \frac{ \xi\mathcal Q}{d-2} \, .
\ee 
Thus, we obtain
\be
\delta M = \frac{1}{16 \pi} \int_\Sigma \xi {\mathcal Q} \, .
\ee
Since $\mathcal Q \geq 0$, if $\mathcal Q \neq 0$ anywhere on $\Sigma$, we have $\delta M > 0$, as we desired to show.
\end{proof}

\section{Existence of a Reflection Isometry} \label{proofsec}

We now have all of the ingredients in place to state and prove the main result of this paper:

\begin{theorem}
Let $(M,g_{ab})$ be a $d$-dimensional solution to Einstein's equation in vacuum representing an asymptotically flat, strongly asymptotically predictable, stationary and axisymmetric black hole with a non-degenerate horizon. Suppose that the action of the stationary-axisymmetric isometry group $G$ is trivial, as defined in the first paragraph of section \ref{assumptions}. Then there exists an isometry $i: M \to M$ that reverses the direction of the timelike Killing field $t^a$ and axial Killing fields ${\phi_A}^a$ (i.e., $i^*t^a = -t^a$ and $i^*{\phi_A}^a = - {\phi_A}^a$) and leaves invariant a $(d-n-1)$-dimensional surface $\mathcal S$ that is orthogonal to $t^a$ and ${\phi_A}^a$. In particular, the metric may be put in a ``block diagonal'' form.
\end{theorem}

\begin{proof}
Let $\Sigma$ be a maximal (${K^i}_i = 0$) Cauchy surface and consider the perturbation $(\delta h_{ij}, \delta K_{ij})$ constructed in the previous section. This perturbation satisfies $\delta A = \delta J_A = 0$. Hence, by the first law of black hole mechanics \eqref{1stlaw}, we have $\delta M = 0$. By the lemma of the previous section, we therefore have $\mathcal Q = 0$, where $\mathcal Q$ was defined by \eqref{Q}. Consequently, the polar parts, $k_{ij}$ and $K_{AB}$, of the extrinsic curvature vanish and we have
\be
{F^A}_{ij} = 0 \, .
\ee
However, ${F^A}_{ij} = (d A^A)_{ij}$, and, as noted in section \ref{manifoldoforbits}, the axis regularity conditions implied by the trivial action of $G$ (see appendix \ref{axisappendix}) ensure that this relation holds everywhere on $\Sigma$, not just on $\widetilde{\Sigma}$. By the topological censorship theorem, $\Sigma$ is simply connected, so there exists a smooth $\mathcal V$-valued function $\chi^A$ such that ${A^A}_i = D_i \chi^A$. Hence, we can set ${A^A}_i = 0$ by choosing a new cross-section $\mathcal S$ by displacing the original cross-section by $-\chi^\Lambda$ along the orbit of each ${\phi_\Lambda}^i$. Equivalently, in terms of the coordinates used in \eqref{metform2}, we redefine $\varphi^\Lambda \to \varphi^\Lambda + \chi^\Lambda$ and thus set $A_{\Gamma \mu} = 0$. This shows that $h_{ij}$ possesses the desired reflection symmetry under $\varphi^\Lambda \to - \varphi^\Lambda$. As shown at the end of section \ref{assumptions} above, the vanishing of the polar part of $K_{ij}$ then proves the existence of the desired $(t$-$\phi)$-reflection isometry. 
\end{proof}

\noindent
{\bf Remarks:} 
\begin{itemize}

\item If the action of $G$ were such that $\widetilde{M}$ corresponded to a nontrivial principal bundle, then the curvature, ${F^A}_{ij}$ must be nonvanishing, and one would not be able to find a $(d-n-1)$-dimensional surface $\mathcal S$ that is orthogonal to $t^a$ and ${\phi_A}^a$. Hence, a ``$(t$-$\phi)$-reflection'' isometry having the properties stated in the theorem cannot exist when the action of $G$ corresponds to a nontrivial principal bundle.

\item  It is interesting to see why our proof breaks down when $G$ corresponds to a nontrivial bundle. The assumption of a trivial action of $G$ was used in the proof to give a global definition of ${A^A}_i$ on $\Sigma$. This allowed us to define the metric perturbation \eqref{pert1}, which has the effect of scaling up ${F^A}_{ij}$. However, if the action of $G$ corresponded to a nontrivial fiber bundle structure on $\widetilde{\Sigma}$, then ${F^A}_{ij}$ would contain a magnetic monopole when viewed as a tensor field on the manifold of orbits $\mathcal O$. Since magnetic monopole charge is ``quantized,'' there cannot exist a perturbation that scales ${F^A}_{ij}$. In this way, the assumption of a trivial action of $G$ was essential to our proof. 

\item If the horizon is non-rotating, i.e., if $\Omega^\Lambda = 0$ for all $\Lambda$, then $\delta J_\Lambda$ does not contribute to the first law \eqref{1stlaw}. Instead of choosing the perturbation \eqref{pert1}, we may choose $\delta'_1 h_{ij} = 0$, $\delta'_1 K_{ij} = K_{ij}$ and then modify this perturbation by a conformal perturbation so as to satisfy all of the constraints. A repetition of the argument used in the proof of the above theorem then shows that $K_{ij} = 0$ and, thus, $t^a$ is hypersurface orthogonal---i.e., the spacetime is static---without the need to assume trivial action of $G$. This is precisely the staticity theorem of Sudarsky and Wald \cite{sudwald}. If, in addition, we assume that the action of $G$ is trivial, then our theorem above also proves that the axial symmetries by themselves (i.e., without inclusion of the time translations) are surface orthogonal.

\item More generally, suppose there is a subgroup, $G' \subset G$, that includes the stationary isometries and all axial Killing fields with non-zero horizon rotation, i.e., suppose that the horizon Killing field lies in the span of the Killing fields of $G'$. Suppose further that $G'$ also has trivial action. Then our theorem also holds with $G$ replaced by $G'$. Hence, we obtain an additional reflection isometry $i'$. The isometry $i \circ i'$ maps the Killing fields in $G'$ to themselves, but reflects a complementary subspace of Killing fields in $G$. The Killing fields in this complementary subspace must be everywhere orthogonal to $t^a$ 
and the axial Killing fields in $G'$. The previous remark is a special case of this remark in the case of a non-rotating black hole, where we can choose $G'$ to be the time translations. 

\item It is worth noting that all of the known (to us) exact black hole vacuum solutions in higher dimensions can be seen to obey the theorem. Specifically, the Myers-Perry \cite{MyersPerry} solutions in $d$-dimensions (which are stationary and axisymmetric with $\left\lfloor\frac{d-1}{2}\right\rfloor$ commuting axial Killing fields) and the doubly-spinning black ring solution \cite{PomeranskySenkov} in $d=5$ dimensions are ($t$-$\phi$)-reflection symmetric. (Of course, the $d=5$ case with two axial Killing fields is also covered by the direct analog of the proof of Papapetrou and Carter, as discussed in section \ref{intro}.) Furthermore, the Myers-Perry \cite{MyersPerry} solutions when one or more of the axial Killing fields have zero horizon rotation and the singly-spinning black ring \cite{EmparanReall2} possess an additional reflection symmetry of the ``non-spinning'' axial Killing field, in accord with the previous remark.

\end{itemize}

\section{Generalizations} \label{gensec}

Our analysis in this paper has been restricted to asymptotically flat black holes that satisfy the vacuum Einstein equation. We conclude by making some remarks on possible generalizations of our results.

Einstein's equation entered our analysis in an essential\footnote{There are some additional places where Einstein's equation was used, such as for obtaining uniqueness of a maximal foliation (although Einstein's equation is not needed for its existence), for topological censorship, and for the specific formulas for ADM mass and angular momentum, but these uses do not appear as essential to the proof as the ones listed.} way in the following places: (1) the validity of the first law of black hole mechanics \eqref{1stlaw}; (2) the fact that the linearized momentum constraints can be solved by a scaling transformation of the form \eqref{pert1}; (3) the fact that the linearized momentum constraints hold for the conformal perturbation \eqref{conformal}; (4) the fact that for the conformal perturbation \eqref{conformal}, the linearized Hamiltonian constraint \eqref{conformalhamiltonianconstraint} can be solved for a given source.

Property (1) is a very general feature of diffeomorphism covariant theories of gravity \cite{IyerWald}. In particular, it will hold for Einstein's equation with a cosmological constant and it will hold if any matter sources derived from a Lagrangian are added. However, the addition of matter sources may result in the presence of additional terms in the first law of black hole mechanics associated with matter field charges. 

Property (2) will hold for Einstein's equation with a cosmological constant. It also will hold in many instances for Einstein's equation with matter sources, provided that the scaling transformation also suitably scales the axial part of the matter configuration variables and polar part of the matter momentum variables. In particular, it will hold for Maxwell fields and, more generally, for Yang-Mills fields. However, the scaling of the polar part of the electric field, $E^a$, will scale the electric charge, which---on account of the charge term in the first law of black hole mechanics---will invalidate the proof of our theorem. Nevertheless, in the case of Maxwell fields---but not in the case of Yang-Mills fields---one can reverse the roles of $E^a$ and the vector potential $A^a$ as done in \cite{sudwald} in order to prove the theorem.

Property (3) will hold for Einstein's equation with a cosmological constant.
It will also hold in many instances for Einstein's equation with matter sources, provided that the conformal scaling for the matter configuration and momentum variables are chosen so that the contribution to the momentum constraint from the matter fields, $\sqrt{h}n^a {h^b}_c T_{ab}$, is unchanged under the conformal perturbation. In particular, it will hold for Maxwell and Yang-Mills fields.

Property (4) will hold for Einstein's equation with a \emph{negative} cosmological constant, as including a cosmological constant $\Lambda$ introduces the term $-2\Lambda \psi/(d-2)$ to the left hand side of \eqref{e2}. It will also hold in many instances for Einstein's equation for matter sources, provided that certain energy properties hold\footnote{For example, Einstein's equation coupled to a Klein-Gordon scalar field satisfies properties (2) and (3), but is guaranteed to satisfy property (4) only if the mass of the Klein-Gordon field is non-positive, as the potential energy contributes to equation \eqref{e2} with the same sign as the cosmological constant term.}. In particular, it will hold for Maxwell and Yang-Mills fields.

Thus, our results should generalize straightforwardly to the Einstein-Maxwell case and to black holes in asymptotically AdS spacetimes. It is likely that additional generalizations can be made.

\section*{Acknowledgments}
We wish to thank Stefan Hollands and Kartik Prabhu for helpful discussions and comments. This research was supported in part by NSF grant PHY 12-02718 to the University of Chicago.

\appendix
\section{Axis Regularity Conditions} \label{axisappendix}

In this Appendix, we derive the regularity conditions that must hold on the axes, $\mathcal A$, under the assumption that the 
isometry group $G=R\times [U(1)]^n$ acts trivially. As explained in section \ref{assumptions}, we are interested in the case where $t^a$ is linearly independent of ${\phi_\Lambda}^a$ (and thus, in particular, is nonvanishing) in the region of interest in $M$ and the coordinate $t$ can be chosen so that the hypersurfaces of constant $t$ are spacelike and axisymmetric (i.e., ${\phi_\Lambda}^a$ are tangent to these surfaces); we therefore assume that these conditions hold. We consider the behavior of ${\phi_\Lambda}^a$ on one of these $(d-1)$-dimensional $t = \text{constant}$ hypersurfaces, $\Sigma$, in a neighborhood of an axis point $p \in \mathcal A$, i.e., $p$ is a point of $\Sigma$ at which the axial Killing fields become linearly dependent. We choose a new basis of axial Killing fields (if necessary) so that $m$ of the Killing fields vanish at $p$ while the other $(n-m)$ Killing fields are linearly independent at $p$. 

First, for simplicity, consider the case $m=1$, and let $\phi^i$ be the Killing field that vanishes at $p$. Consider the linear map on the tangent space to $\Sigma$ at $p$ given by
\be
{L^i}_j \equiv \left. D^i \phi_j\right|_p \, .
\ee
Since $D_i \phi_j$ is antisymmetric by Killing's equation, it follows that the linear map $i {L^i}_j$ is self-adjoint, and can therefore be diagonalized by an orthonormal basis of (complex) eigenvectors with real eigenvalues. These eigenvectors come in complex conjugate pairs, i.e., if $Z^i$ is an eigenvector with nonzero eigenvalue $\lambda$, then $\bar Z^i$ is an eigenvector with eigenvalue $-\lambda$. Thus, we can write
\be
i {L^i}_j = \sum_{\alpha=1}^{r}\lambda_\alpha \left( {Z_\alpha}^i \bar Z_{\alpha j} - \bar Z_\alpha^{\phantom\alpha i} Z_{\alpha j} \right) \, ,
\label{Lform1}
\ee
where each $\lambda_\alpha$ in this expansion is nonzero. Note that at least one term must occur in this sum, since if both $\phi^i$ and its derivative vanished at $p$, then $\phi^i$ would vanish identically. Thus, we have $r\geq1$. On the other hand, we have $r\leq \left\lfloor\frac{d-1}{2}\right\rfloor$, since we cannot have more than $\left\lfloor\frac{d-1}{2}\right\rfloor$ pairs of orthogonal vectors on a $(d-1)$-dimensional space. In particular, we automatically have $r=1$ in $d=4$ spacetime dimensions.

We now show that if more than one term occurs---i.e., if $r > 1$---then the orbits of $\phi^i$ near $p$ are homotopic to a point, and thus the group action cannot be trivial. Writing
\be
{Z_{\alpha}}^i = \frac{1}{\sqrt 2} \left( {x_{\alpha}}^i + i {y_{\alpha}}^i\right) \, ,
\ee
where ${x_1}^i, {y_1}^i, \dots, {x_r}^i, {y_r}^i$ are real orthonormal vectors in the tangent space to $\Sigma$ at $p$, we obtain
\be
L_{ij} = \sum_{\alpha=1}^{r}\lambda_\alpha \left(y_\alpha \wedge x_\alpha \right)_{ij}  \,  .
\ee
Thus, in the tangent space to $\Sigma$ at $p$, ${L^i}_j$ can be recognized as the generator of simultaneous rotations in the 2-planes spanned by $\left({x_\alpha}^i,{y_\alpha}^i\right)$. 
We can choose Euclidean coordinates 
\be\label{coords}
\left( x^1, y^1, \dots, x^r, y^r, z^1, \dots, z^{(d-1-2r)} \right)
\ee
on the tangent space to $\Sigma$ at $p$ such that
\be
{x_\alpha}^i = \left(\frac{\partial}{\partial x^\alpha}\right)^i , \qquad {y_\alpha}^i = \left(\frac{\partial}{\partial y^\alpha}\right)^i  \, ,
\ee
and promote these to Riemann normal coordinates on $\Sigma$ in a neighborhood, $\mathcal U$, of $p$. Within $\mathcal U$, the Killing field $\phi^i$ vanishes precisely on the surface $x_1=y_1=\dots=x_r=y_r=0$. If $\mathcal U$ is chosen to be sufficiently small, no other Killing field can vanish in $\mathcal U$. Thus, there exists a neighborhood, $\mathcal U$, of $p$ such that $\mathcal U \cap \widetilde{\Sigma}$ has the topology of $R^{(d-1)}$ with a plane of dimension $(d-1)-2r$ removed. If $r>1$, then $\mathcal U \cap \widetilde{\Sigma}$ is simply connected, so the Killing orbits of $\phi^i$ near $p$ are homotopic to a point, contradicting our assumption of trivial $[U(1)]^n$ action. Thus we must have $r=1$. Note that the Sorkin monopole (or the metric of eq. \eqref{sorkinmonopole}) corresponds to the case $r=2$ in $d=5$ dimensions. 

Thus, we have shown that a necessary condition for a trivial group action is that $L_{ij}=\left.D_i \phi_j \right|_p$ be a simple bivector, $L_{ij}=\lambda \left( y\wedge x\right)_{ij}$. We now derive axis regularity conditions in this case. First, we rescale $\phi^i$ if necessary via $\phi^i \to \phi^i/\lambda$ so that
\be
L_{ij}= \left( y\wedge x\right)_{ij} \, .
\ee
The action of the Killing field $\phi^i$ on the tangent space at $p$ is simply to perform a rotation in the $x^i$-$y^i$ plane, leaving invariant any vector ${z_\mu}^i$ that is orthogonal to $x^i$ and $y^i$. It follows that in the Riemannian normal coordinates $(x,y,z^\mu)$ constructed above, the Killing field $\phi^i$ takes the form
\be
\phi^i = x \left(\frac{\partial}{\partial y}\right)^i - y \left(\frac{\partial}{\partial x}\right)^i \, .
\ee
Define the polar coordinates $\varphi$ and $\rho$ in the Riemannian normal coordinate neighborhood $\mathcal U$ by
\be
\tan \varphi = y/x \, , \quad \quad \rho^2 = x^2 + y^2 \, .
\ee
Now extend the coordinates $\left( \varphi,\rho,z^\mu\right)$ off of $\Sigma$ by Killing transport along $t^a$. Then, in the coordinates $(t,\varphi, x^\mu)$ with $x^\mu = \{\rho, z^\nu \}$, the metric will take the form \eqref{metricform}. If we translate the smoothness of the components of the metric in the coordinates $(t,x,y,z^\mu)$ into corresponding conditions on the metric components in the coordinates $(t, \phi, x^\mu)$ as one approaches the axis $\rho=0$, one obtains the following conditions: (i) On any $t = {\rm const}$, $\phi = {\rm const}$ surface, the metric $\lambda_{\mu \nu}$ smoothly extends to $\rho=0$. (ii) On $M$, the scalar field $\beta/\Phi$ and one-form field $A_\mu dx^\mu/\Phi$ smoothly extend to $\rho=0$, where $\Phi = \phi^i \phi_i$. (iii) On $M$, the one-form field $B_\mu dx^\mu$ and the tensor field $\Phi d\phi^2 + \lambda_{\mu \nu} dx^\mu dx^\nu$ smoothly extend\footnote{Note that neither of the individual tensor fields $\Phi d\phi^2$ and  $\lambda_{\mu \nu} dx^\mu dx^\nu$ extend smoothly to $\rho=0$ on $\Sigma$ (although $\lambda_{\mu \nu} dx^\mu dx^\nu$ extends smoothly to $\rho=0$ when restricted to any $t, \phi = {\rm const}$ surface).} to $\rho=0$.

The general case where $1\leq m \leq n$ can be analyzed similarly. Suppose that ${\phi_\Lambda}^i$ for $\Lambda=1,\dots,m$ vanish at $p$. Let
\be
{{L_\Lambda}^i}_j \equiv \left. D^i \phi_{\Lambda j}\right|_p \, .
\ee
Since ${\phi_\Lambda}^i$ for $\Lambda=1,\dots,m$ are linearly independent as Killing fields, the linear maps ${{L_\Lambda}^i}_j$, on the tangent space to $\Sigma$ at $p$ must be linearly independent.
Since the ${\phi_\Lambda}^i$ commute, the linear maps ${{L_\Lambda}^i}_j$ must commute because
\ba
0 = \left.D_j\left[ \phi_\Lambda, \phi_\Sigma\right]^i \right|_p &= \left.D_j\left({\phi_\Lambda}^k D_k \phi_\Sigma^i - {\phi_\Sigma}^k D_k \phi_\Lambda^i\right)\right|_p \\
&= \left.\left((D_j{\phi_\Lambda}^k)( D_k \phi_\Sigma^i) -( D_j{\phi_\Sigma}^k)( D_k \phi_\Lambda^i)\right)\right|_p = {{\left[ L_\Sigma , L_\Lambda \right]}^i}_j \, .
\ea
Thus, the collection of linear maps $i {{L_\Lambda}^i}_j$  for $\Lambda=1,\dots,m$ are self-adjoint and commuting, so they can be simultaneously diagonalized by an orthonormal basis of complex eigenvectors with real eigenvalues, again coming in complex conjugate pairs. In parallel with \eqref{Lform1}, we have
\be
i {{L_\Lambda}^i}_j = \sum_{\alpha=1}^{r}\lambda_{\Lambda\alpha} \left( {Z_\alpha}^i \bar Z_{\alpha j} - \bar Z_\alpha^{\phantom\alpha i} Z_{\alpha j} \right)  \, ,
\ee
where the linear independence of the collection ${{L_\Lambda}^i}_j$ requires that $\lambda_{\Lambda\alpha}$ be an $m\times r$ (real) matrix with linearly independent rows. In particular, we have $r\geq m$, so that now $m\leq r\leq \left\lfloor\frac{d-1}{2}\right\rfloor$. Again, we may write
\be
{Z_{\alpha}}^i = \frac{1}{\sqrt 2} \left( {x_{\alpha}}^i + i {y_{\alpha}}^i\right) \, ,
\ee
where ${x_1}^i, {y_1}^i, \dots, {x_r}^i, {y_r}^i$ are real orthonormal vectors in the tangent space to $\Sigma$ at $p$. We obtain
\be
L_{\Lambda ij}  = \sum_{\alpha=1}^{r}\lambda_{\Lambda\alpha} \left(y_\alpha \wedge x_\alpha \right)_{ij}  \,  .
\ee

Consider, first, the case $r>m$. Then there exists at least one Killing field $\phi^i$ in the collection whose derivative, $L_{ij}$, at $p$ is not a simple bivector, and such that no linear combination of the other Killing fields in the collection can be added to $\phi^i$ so as to make its derivative at $p$ be a simple bivector. By arguments similar to the above case where $m = 1$, one can then show that the orbits of $\phi^i$ in a neighborhood of $p$ are homotopic to a point. Thus, for a trivial group action, the case $r > m$ cannot occur.

In the case $r=m$, we can find a basis of ${\phi_\Lambda}^i$ where each $L_{\Lambda ij}$ is a simple bivector (i.e., we can diagonalize $\lambda_{\Lambda\alpha}$). Furthermore, we can rescale the Killing fields so that each bivector has unit coefficient, i.e., we have
\be
L_{\Lambda ij}  = \left(y_\Lambda \wedge x_\Lambda \right)_{ij} \, ,
\ee
where it is understood that no sum over $\Lambda$ is taken.
In Riemannian normal coordinates $({x^\Lambda}, {y^\Lambda}, {z^\mu})$ associated with the orthonormal basis $({x_\Lambda}^i, {y_\Lambda}^i, {z_\mu}^i)$, the Killing fields ${\phi_\Lambda}^i$ take the form
\be
{\phi_\Lambda}^i = x^\Lambda \left(\frac{\partial}{\partial y^\Lambda}\right)^i - y^\Lambda \left(\frac{\partial}{\partial x^\Lambda}\right)^i \, ,
\ee
where, again, no sum over $\Lambda$ is taken. We can again define polar coordinates
\be
\tan \phi^\Lambda = y^\Lambda/x^\Lambda \, , \quad \quad \left(\rho^\Lambda\right)^2 = \left(x^\Lambda\right)^2 + \left(y^\Lambda\right)^2 \, ,
\ee
so as to put the metric in the form \eqref{metricform}. The smoothness of the metric components in the coordinates $(t,x^\Lambda,y^\Lambda,z^\mu)$ then imply the axis regularity conditions stated in 
section \ref{assumptions}.

\bibliography{mybib}

\end{document}